%% file: paper.tex
\documentclass[]{llncs}
\pdfoutput=1 
\usepackage[dvips]{graphicx}
\usepackage[english]{babel}
\usepackage{amssymb,amsfonts,amsmath}
\usepackage{color}
\usepackage{xspace}
\usepackage{url}
\usepackage[usenames,dvipsnames]{xcolor}
\usepackage{numprint}
\usepackage{multirow}
\usepackage[algo2e,ruled,vlined]{algorithm2e}
\npdecimalsign{.} 

\usepackage{tabulary}
\usepackage{longtable}
\usepackage{booktabs}
\usepackage{float}
\usepackage{placeins}

\let\doendproof\endproof
\renewcommand\endproof{~\hfill$\boxtimes$\doendproof}

\def\comment#1{}

\def\withcomments{
  \newcounter{mycommentcounter}
   \def\comment##1{\refstepcounter{mycommentcounter}%
    \ifhmode%
     \unskip%
     {\dimen1=\baselineskip \divide\dimen1 by 2 %
       \raise\dimen1\llap{\tiny\bfseries \textcolor{red}{-\themycommentcounter-}}}\fi%
     \marginpar[{\renewcommand{\baselinestretch}{0.8}%
       \hspace*{3em}\begin{minipage}{5em}\footnotesize [\themycommentcounter]: \raggedright ##1\end{minipage}}]{\renewcommand{\baselinestretch}{0.8}%
       \begin{minipage}{5em}\footnotesize [\themycommentcounter]: \raggedright ##1\end{minipage}}}
  }

\definecolor{orange}{RGB}{255,80,0}

\newcommand{\strash}[1]{{\color{orange}[DS: #1]}}
\renewcommand{\strash}[1]{}

\usepackage{hyperref}

\newcommand{\AlgName}[1]{\ensuremath{\text{{\sf #1}}}}
\newcommand{\advanced}{\AlgName{Advanced}}
\newcommand{\simple}{\AlgName{Simple}}
\newcommand{\critical}{\AlgName{Critical}}
\newcommand{\maxcritical}{\AlgName{MaxCritical}}
\newcommand{\ai}{\AlgName{B\&R}}
\newcommand{\mcs}{\AlgName{MCS}}

\title{On the Power of Simple Reductions for the Maximum Independent Set Problem}

\author
{
  Darren Strash
}
\institute{Institute of Theoretical Informatics \\ Karlsruhe Institute of Technology \\ Karlsruhe, Germany \\ \email{strash@kit.edu}}

\begin{document}

\maketitle

\begin{abstract}
Reductions---rules that reduce input size while maintaining the ability to compute an optimal solution---are critical for developing efficient maximum independent set algorithms in both theory and practice. While several simple reductions have previously been shown to make small domain-specific instances tractable in practice, it was only recently shown that advanced reductions (in a measure-and-conquer approach) can be used to solve real-world networks on millions of vertices [Akiba and Iwata, TCS 2016]. In this paper we compare these state-of-the-art reductions against a small suite of simple reductions, and come to two conclusions: just two simple reductions---vertex folding and isolated vertex removal---are sufficient for many real-world instances, and further, the power of the advanced rules comes largely from their initial application (i.e., kernelization), and not their repeated application during branch-and-bound. As a part of our comparison, we give the first experimental evaluation of a reduction based on maximum critical independent sets, and show it is highly effective in practice for medium-sized networks.
\end{abstract}


\section{Introduction}
Given a graph $G=(V,E)$, the maximum independent set problem asks us to compute a maximum cardinality set of vertices $I\subseteq V$ such that no vertices in $I$ are adjacent to one another. Such a set is called a \emph{maximum independent set} (MIS).
The maximum independent set problem has applications in classification theory, information retrieval, computer vision~\cite{feo1994greedy}, computer graphics~\cite{sander-mesh-2008}, map labeling~\cite{gemsa2014dynamiclabel,verweig-map-labelling-1999} and routing in road networks~\cite{kieritz-contraction-2010}, to name a few.  However, the maximum independent set problem is NP hard~\cite{garey-johnson-1979}, and therefore, the currently-best-known algorithms take exponential time.

\subsection{Previous Work}
Most previous work has focused on the \emph{maximum clique} problem and the \emph{minimum vertex cover} problem, which are complementary to ours. That is, the maximum clique in the complement graph $\bar{G}$ is a maximum independent set in $G$, and if $C$ is a minimum vertex cover in G, then $V\setminus C$ is a maximum independent set.  

For computing a maximum clique, there are many branch-and-bound algorithms that are efficient in practice~\cite{segundo-recoloring,segundo-bitboard-2011,tomita-recoloring}. These algorithms achieve fast running times by prescribing the order to select vertices during search and by implementing fast-but-effective pruning techniques, such as those based on approximate graph coloring~\cite{tomita-recoloring} or MaxSAT~\cite{li-maxsat-2013}. Among the fastest of these algorithms is the \mcs{} algorithm by Tomita et al.~\cite{tomita-recoloring}, which is competitive in dense graphs even against algorithms that use bit parallelism~\cite{segundo-recoloring}. Further priming these algorithms with a large initial solution obtained using local search~\cite{andrade-2012} can be surprisingly effective at speeding up search~\cite{batsyn-mcs-ils-2014}.

Several techniques based on kernelization~\cite{abu-khzam-2007,bodlaender-kernelization-2013} have been very promising in solving both the maximum independent set and minimum vertex cover problems. In particular, Butenko et al.~\cite{butenko-correcting-codes-2009} showed that isolated vertex reductions disconnect medium-sized graphs derived from error-correcting codes into small connected components that can be solved optimally. Butenko and Trukhanov~\cite{butenko-trukhanov} introduced a reduction based on critical independent sets, finding exact maximum independent sets in graphs with up to 18,000 vertices generated with the Sanchis graph generator~\cite{sanchis-1996}. Though these works apply reduction techniques as a preprocessing step, further works apply reductions as a natural step of the algorithm. In the area of exact algorithms, it has long been clear that applying reductions in a measure-and-conquer approach can improve the theoretical running time of vertex cover and independent set algorithms~\cite{bourgeois-reduction-theory-2012,fomin-2010}. However, few experiments have been conducted on the real-world efficacy of these techniques.

Recently, Akiba and Iwata~\cite{akiba-tcs-2016} showed that applying advanced reductions with sophisticated branching rules in a measure-and-conquer approach is highly effective in practice. They show that an exact minimum vertex cover, and therefore an exact maximum independent set, can be found in many large complex networks with up to 3.2 million vertices in much less than a second. Further, on nearly all of their inputs, the state-of-the-art branch-and-bound algorithm \mcs~\cite{tomita-recoloring} fails to finish within 24 hours. Thus, their method is orders of magnitude faster on these real-world graphs.

\subsection{Our Results}
While the results of Akiba and Iwata~\cite{akiba-tcs-2016} are impressive, it is not clear how much their advanced techniques actually improve \emph{search} compared to existing techniques. A majority of the graphs they tested have kernel size zero, and therefore no branching is required. We show that just 2 simple reduction rules---isolated vertex removal and vertex folding---are sufficient to make many of their test instances tractable with standard branch-and-bound solvers. We further provide the first comparison with another class of reductions that are effective on real-world complex networks: the critical independent set reduction of Butenko and Trukhanov~\cite{butenko-trukhanov} and the variant due to Larson~\cite{larson-2007}, which computes a maximum critical independent set.

\section{Preliminaries}
We work with an undirected graph $G = (V,E)$ where $V$ is a set of $n$ vertices and $E\subset \{\{u,v\}\mid u,v\in V\}$ is a set of $m$ edges. The open neighborhood of a vertex $v$, denoted $N(v)$, is the set of all vertices $w$ such that $(v,w)\in E$. We further denote the closed neighborhood by $N[v]=N(v)\cup\{v\}$. We similarly define the open and closed neighborhoods of a set of vertices $U$ to be $N(U) = \bigcup_{u\in U}N(u)$ and $N[U] = N(U) \cup U$, respectively. Lastly, for vertices $S\subseteq V$, the induced subgraph $G[S]\subseteq G$ is the graph on the vertices in $S$ with edges in $E$ between vertices in $S$.

\subsection{Reduction Rules}
There are several well-known reduction rules that can be applied to graphs for the minimum vertex cover problem (and hence the maximum independent set problem) to reduce the input size to its irreducible equivalent, the \emph{kernel}~\cite{abu-khzam-2007}. 
Each reduction allows us to choose vertices that are in some MIS by following simple rules. If an MIS is found in the kernel, then undoing the reductions gives an MIS in the original graph.
Reduction rules are typically applied as a preprocessing step. The hope is that the kernel is \emph{small enough} to be solved by existing solvers in feasible time. If the kernel is empty, then a maximum independent set is found by simply undoing the reductions.  
We now briefly describe three classes of reduction rules that we consider here.

\subsubsection{Simple Reductions.}

We first describe two simple reductions:  \emph{isolated vertex removal} and \emph{vertex folding}.

An isolated vertex, also called a \emph{simplicial vertex}, is a vertex $v$ whose neighborhood forms a clique. That is, there is a clique $C$ such that $V(C) \cap N(v) = N(v)$. Since $v$ has no neighbors outside of the clique, it must be in \emph{some} maximum independent set. Therefore, we can add $v$ to the maximum independent set we are computing, and remove $v$ and $C$ from the graph. Isolated vertex removal was shown by Butenko et al.~\cite{butenko-correcting-codes-2009} to be highly effective in finding exact maximum independent sets on graphs derived from error-correcting codes~\cite{butenko-correcting-codes-2009}. This reduction is typically restricted to vertices of degree zero, one, and two in the literature. However, we consider vertices of any degree.

Vertex folding was first introduced by Chen et al.~\cite{chen-1999} to reduce the theoretical running time of exact branch-and-bound algorithms for the maximum independent set problem. This reduction is applied whenever there is a vertex $v$ with degree 2 and non-adjacent neighbors $u$ and $w$. Either $v$ or both $u$ and $w$ are in some MIS. Therefore, we can contract $u$, $v$, and $w$ to a single vertex $v^{\prime}$ and add the appropriate vertices to the MIS after finding an MIS in the kernel.

\subsubsection{Critical Independent Set Reductions.}
One further reduction method shown to be effective in practice for sparse graphs is the critical independent set reduction by Butenko and Trukhanov~\cite{butenko-trukhanov}. A \emph{critical set} is a set $U\subseteq V$ that maximizes $|U| - |N(U)|$ and $I_c = U\setminus N(U)$ is called a \emph{critical independent set}. Butenko and Trukhanov show that every critical independent set is contained in some maximum independent set, and show that one can be found in polynomial time. Their algorithm works by repeatedly computing some critical independent set $I_c$ and removing $N[I_c]$ from the graph, stopping when $I_c$ is empty.

\begin{figure}[tb]
\centering
\includegraphics[]{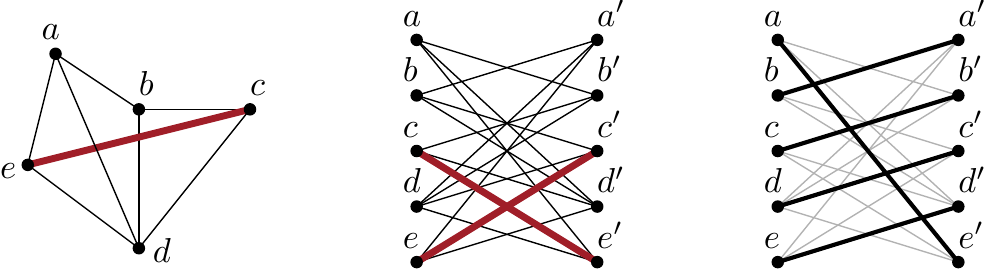}
\caption{A graph $G$ (left) and its bi-double graph $B(G)$ (middle), illustrating that edges of $G$ become two edges in $B(G)$. Right: a maximum matching (in this instance, a perfect matching) in $B(G)$.}
\label{figure:bi-double}
\vspace*{-0.3cm}
\end{figure}

 A critical set can be found by first computing the \emph{bi-double graph} $B(G)$, then computing a maximum independent set in $B(G)$~\cite{ageev-1994,larson-2007,zhang-1990}. $B(G)$ is a bipartite graph with vertices $V\cup V^{\prime}$ where $V'$ is a copy of $V$, and contains edge $(u,v^\prime)\subseteq V\times V^{\prime}$ if and only if $(u,v)\in E$. Since $B(G)$ is bipartite, the maximum independent set in $B(G)$ can be solved by computing a maximum bipartite matching in polynomial time~\cite{hopcroft-karp}.  

Butenko and Trukhanov~\cite{butenko-trukhanov} use the standard augmenting path algorithm to compute a maximum independent set in $B(G)$, and hence find a critical set in $G$, in $O(nm)$ time.
One drawback of their approach is that the quality of the reduction depends on the maximum independent set found in the bi-double graph. As noted by Larson~\cite{larson-2007}, if there is a perfect matching in $B(G)$ (such as in Fig.~\ref{figure:bi-double}, right), then $G$ has an empty critical empty set. However, in the experiments by Butenko and Trukhanov~\cite{butenko-trukhanov} these worst cases were not observed.

To prevent the worst-case, Larson~\cite{larson-2007} gave the first algorithm to find a maximum critical independent set, which accumulates vertices that are in \emph{some} critical independent set and excludes their neighbors. He further gave a simple method to test if a vertex $v$ is in a critical independent set: vertex $v$ is in a critical independent set if and only if $\alpha{(B(G))} = \alpha{(B(G) - \{v,v^{\prime}\} - N(\{v,v{^\prime}\}))} + 2$, where $\alpha(\cdot)$ is the \emph{independence number}---the size of a maximum independent set. A naive approach would compute a new maximum matching from scratch to compute the independence number of each such bi-double graph, taking $O(n^2m)$ time total (or $O(n^{3/2}m)$ time with the Hopcroft--Karp algorithm~\cite{hopcroft-karp}). However, we can save the matching between executions to ensure only few augmenting paths are computed for each subsequent matching, giving $O(m^2)$ running time, which is better when $m = o(n^{3/2})$.
See Appendix~\ref{section:matchings} for further details.

\subsubsection{Advanced Reduction Rules.}
We list the advanced reduction rules from Akiba and Iwata~\cite{akiba-tcs-2016}. Refer to Akiba and Iwata~\cite{akiba-tcs-2016} for a more thorough discussion, including implementation details.

Firstly, they use vertex folding and degree-1 isolated vertex removal (also called \emph{pendant} vertex removal), as described previously. They further test a full suite of other reductions from the literature, which we now briefly describe.

\noindent\emph{Linear Programming:}
A well-known~\cite{nemhauser-1975} linear programming relaxation for the MIS problem with a half-integral solution (i.e., using only values 0, 1/2, and 1) can be solved using bipartite matching: maximize $\sum_{v\in V}{x_v}$ such that $\forall (u, v) \in E$, $x_u + x_v \leq 1$ and $\forall v \in V$, $x_v \geq 0$. Vertices with value 1 must be in some MIS and can thus be removed from $G$ along with their neighbors. Akiba and Iwata~\cite{akiba-tcs-2016} compute  a solution whose half-integral part is minimal~\cite{iwata-2014}.

\noindent\emph{Unconfined~\cite{Xiao201392}:} Though there are several definitions of \emph{unconfined} vertex in the literature, we use the simple one from Akiba and Iwata~\cite{akiba-tcs-2016}. A vertex $v$ is \emph{unconfined} when determined by the following simple algorithm. First, initialize $S = \{v\}$. Then find a $u \in N(S)$ such that $|N(u) \cap S| = 1$ and $|N(u) \setminus N[S]|$ is minimized. If there is no such vertex, then $v$ is confined. If $N(u) \setminus N[S] = \emptyset$, then $v$ is unconfined.  If $N(u)\setminus N[S]$ is a single vertex $w$, then add $w$ to $S$ and repeat the algorithm. Otherwise, $v$ is confined. Unconfined vertices can be removed from the graph, since there always exists an MIS $I$ with no unconfined vertices.

\noindent\emph{Twin~\cite{Xiao201392}:} Let $u$ and $v$ be vertices of degree 3 with $N(u) = N(v)$. If $G[N(u)]$ has edges, then add $u$ and $v$ to $I$ and remove $u$, $v$, $N(u)$, $N(v)$ from $G$. Otherwise, some vertices in $N(u)$ may belong to some MIS $I$. We still remove $u$, $v$, $N(u)$ and $N(v)$ from $G$, and add a new gadget vertex $w$ to $G$ with edges to $u$'s two-neighborhood (vertices at a distance 2 from $u$). If $w$ is in the computed MIS, then none of $u$'s two-neighbors are $I$, and therefore $N(u) \subseteq I$. Otherwise, if $w$ is not in the computed MIS, then some of $u$'s two-neighbors are in $I$, and therefore $u$ and $v$ are added to $I$.

\noindent\emph{Alternative:} Two sets of vertices $A$ and $B$ are set to be \emph{alternatives} if $|A| = |B| \geq 1$ and there exists an MIS $I$ such that $I\cap(A\cup B)$ is either $A$ or $B$. Then we remove $A$ and $B$ and $C = N(A)\cap N(B)$ from $G$ and add edges from each $a \in N(A)\setminus C$ to each $b\in N(B)\setminus C$.
Then we add either $A$ or $B$ to $I$, depending on which neighborhood has vertices in $I$. Two structures are detected as alternatives. First, if $N(v)\setminus \{u\}$ induces a complete graph, then $\{u\}$ and $\{v\}$ are alternatives (a \emph{funnel}). Next, if there is a chordless 4-cycle $a_1b_1a_2b_2$ where each vertex has at least degree 3. Then sets $A=\{a_1, a_2\}$ and $B=\{b_1, b_2\}$ are alternatives when $|N(A) \setminus B| \leq 2$, $|N(A) \setminus B| \leq 2$, and $N(A) \cap N(B) = \emptyset$.

\noindent\emph{Packing~\cite{akiba-tcs-2016}:} 
Given a non-empty set of vertices $S$, we may specify a \emph{packing constraint} $\sum_{v\in S}x_v \leq k$, where $x_v$ is 0 when $v$ is in some MIS $I$ and 1 otherwise. Whenever a vertex $v$ is excluded from $I$ (i.e., in the unconfined reduction), we remove $x_v$ from the packing constraint and decrease the upper bound of the constraint by one. Initially, packing constraints are created whenever a vertex $v$ is excluded from or included in the MIS. The simplest case for the packing reduction is when $k$ is zero: all vertices must be in $I$ to satisfy the constraint. Thus, if there is no edge in $G[S]$, $S$ may be added to $I$, and $S$ and $N(S)$ are removed from $G$. Other cases are much more complex. Whenever packing reductions are applied, existing packing constraints are updated and new ones are added.

\section{Experimental Results}
We first investigate the size of kernels computed by all kernelization techniques. We test four techniques: (1) using only isolated vertex removal and vertex folding (\simple), (2) using the critical independent set reduction rule due to Butenko and Trukhanov~\cite{butenko-trukhanov} (\critical), (3) the version of (2) by Larson~\cite{larson-2007} that always computes a maximum critical independent set (\maxcritical), and (4) the reductions tested by Akiba and Iwata~\cite{akiba-tcs-2016} (\advanced). Note that we use the standard augmenting paths algorithm for computing a maximum bipartite matching (and not the Hopcroft--Karp algorithm~\cite{hopcroft-karp}) to be consistent with the original experiments by Butenko and Trukhanov~\cite{butenko-trukhanov}.

Next, we investigate the time to compute an exact solution on large instances. We test two algorithms: the full branch-and-reduce algorithm due to Akiba and Iwata~\cite{akiba-tcs-2016} (\ai), and \simple{} kernelization followed by \mcs, a state-of-the-art clique solver due to Tomita et al.~\cite{tomita-recoloring} (\simple+\mcs). We use our own implementation of \iftrue \mcs\footnote{\url{https://github.com/darrenstrash/open-mcs}}, \else \mcs, \fi{} since the code for the original implementation is not available and because we modify the MCS algorithm to solve the maximum independent set problem.
We choose MCS because it is one of the leading solvers in practice, even competing with the bit-board implementations of San Segundo et al.~\cite{segundo-bitboard-2011,segundo-recoloring}.

\subsubsection{Instances.}
We run our algorithms on synthetically-generated graphs, as well as a large corpus of real-world sparse data sets. For synthetic cases, we use graphs generated with the Sanchis graph generator~\cite{sanchis-1996}. For medium-sized real-world graphs, we consider small Erd\H{o}s co-authorship networks from the Pajek data set~\cite{pajek} and biological networks from the Biological General Repository for Interaction Datasets v3.3.112 (BioGRID)~\cite{biogrid-2006}. We further consider large complex networks (including co-authorship networks, road networks, social networks, peer-to-peer networks, and Web crawl graphs) from the Koblenz Network Collection (KONECT)~\cite{konect}, the Stanford Large Network Dataset Repository (SNAP)~\cite{snap}, and the Laboratory for Web Algorithmics (LAW)~\cite{boldi-2011,boldi-2004}.

\subsection{Experimental Setup}
All of our experiments were exclusively run on a machine with Ubuntu 14.04.3 and Linux kernel version 3.13.0-77. The machine has four Octa-Core Intel Xeon E5-4640 processors running at 2.4GHz, 512 GB local memory, 4×20 MB L3-Cache, and 4×8×256 KB L2-Cache. For \advanced{} reductions as well as \ai, we compiled and ran the original Java implementation of Akiba and Iwata\footnote{\url{https://github.com/wata-orz/vertex_cover}}~\cite{akiba-tcs-2016} with Java 8 update 60. We implemented all other algorithms\footnote{\url{https://github.com/darrenstrash/kernel-mis}} in C++11, and compiled them with gcc version 4.8.4 with optimization flag \texttt{-O2}. Each algorithm was run for one hour. All running times listed in our tables are in seconds, and we mark a data set with `-' when an algorithm does not finish within the time limit. We indicate the best solution, and the time to achieve it, by marking the value \textbf{bold}.

\input{tables/sanchis.3600.table}
\input{tables/erdos-biogrid.3600.table}

\subsection{Kernel Sizes}

First, we compare the kernel sizes computed by each reduction technique. We first run all algorithms on synthetically generated graphs, on which \critical{} was previously shown to be effective~\cite{butenko-trukhanov}. We generate instances with a known clique number using the Sanchis graph generator\footnote{\url{ftp://dimacs.rutgers.edu/pub/challenge/}}~\cite{sanchis-1996}, and then take the complement. Like Butenko and Trukhanov~\cite{butenko-trukhanov}, we choose the clique number (and thus, the independence number in the complement graph) to be at least $n/2$.

As can be seen in Table~\ref{table:kernel-sanchis}, \critical{} succeeds in reducing the kernel to empty in many cases. During testing, we noticed that \critical{} did not enter a second iteration on most graphs. That is, in general, either the first critical independent set matched the size of a maximum independent set or the remaining graph had an empty critical set. It is unclear what causes this behavior, but we conjecture it could be due to how we compute the maximum matching: we use depth-first search in the bi-double graph. It is unclear which search strategy Butenko and Trukhanov~\cite{butenko-trukhanov} use in their experiments.  
\maxcritical{} always computes an empty kernel on these instances; however, it is significantly slower than \critical. This is because \critical{} computes only 2 maximum matchings on typical instances, while \maxcritical{} computes many more.

We now turn our attention to \simple{} and \advanced{}. \advanced{} is the clear winner on the Sanchis graphs. It always computes an empty kernel, and does so quickly. However, \simple{} also computes empty kernels on 28 of the instances. Even though \simple{} is only faster than \advanced{} on four instances, \simple{} still computes exact solutions on these instances within a few seconds. Therefore, the \advanced{} reductions are not required to make these instances tractable.

We further tested all algorithms on medium-sized real-world graphs. We ran all four reduction algorithms on Erd\H{o}s collaboration graphs from the Pajek data set and on biological graphs from the BioGRID data set (we only show results on those graphs with 100 or more vertices). As seen in Table~\ref{table:kernel-erdos-biogrid}, \maxcritical{} still gives consistently smaller kernels than \critical, but unlike the Sanchis graphs, not all kernels are empty. However, the \simple{} reductions give consistently small kernels on these real-world instances, computing an empty kernel on all but 4 instances, and doing so as fast as \advanced. The size of three of these non-empty kernels is well within the range of feasibility of existing MIS solvers. Neither \simple{} nor \advanced{} can solve the \texttt{Saccharomyces-cerevisiae} data set exactly, and the \simple{} kernel is within 12\% of the \advanced{} kernel size.

\input{tables/large.networks.3600.table}

\subsection{Exact Solutions on Large-Scale Complex Networks.}
We now focus on computing an exact MIS for the larger instances considered by Akiba and Iwata~\cite{akiba-tcs-2016}. We test the branch-and-reduce algorithm by Akiba and Iwata (\ai) that uses \advanced{} reductions with branching rules during recursion. We also run \simple{} to kernelize the graph, and then run \mcs{} on the remaining connected components (\simple+\mcs). Results are presented in Table~\ref{table:solution-large-networks}. Since \critical{} and \maxcritical{} are slow on large instances and less effective on medium-sized real-world instances (see Table~\ref{table:kernel-erdos-biogrid}), we exclude them from these experiments.

Similar to the original experiments of Akiba and Iwata~\cite{akiba-tcs-2016}, \ai{} computes an exact MIS on 42 of these instances. However, surprisingly, \simple+\mcs{} also computes exact solutions for 33 of these instances. In the remaining nine instances where \ai{} computes a solution but \simple+\mcs{} does not, we see that the size $k_{\max}$ of the maximum connected component in the kernel is significantly smaller for \ai{}. For six of these instances, $k_{\max}$ is less than $600$, which is within the range of traditional solvers. Therefore, we conclude that the speed of \ai{} is primarily due to the initial kernelization on these instances. However, the remaining three instances---\texttt{web-BerkStan}, \texttt{web-NotreDame} and \texttt{libimseti}---have kernels that are too large for traditional solvers. Therefore, these instances benefit the most from the branch-and-reduce paradigm.

\section{Conclusion and Future Work}
Although efficient in practice, the techniques used by Akiba and Iwata~\cite{akiba-tcs-2016} are not necessary for computing a maximum independent set exactly in many large complex networks. Our results further suggest that the initial kernelization is far more effective than the techniques used in branch-and-bound. Further, while the critical independent set reduction due to Butenko and Trukhanov~\cite{butenko-trukhanov} and the variant due to Larson~\cite{larson-2007} compute small kernels in practice, they are too slow to compete with other reductions on real-world sparse graphs.

This leaves several open questions that are interesting for future research. In particular, we would like to understand the structure that causes branch-and-reduce techniques to be fast on some graphs, but slow on other (similar) instances. Is it possible to speed up branch-and-reduce algorithms by applying only simple kernelization techniques, and reserving advanced techniques for ``difficult'' portions of the graph? As we've seen, advanced rules are not always necessary. Finally, much time is devoted to computing reductions in branch-and-reduce algorithms, perhaps more advanced (but slower) pruning techniques are now viable for these algorithms. 

\FloatBarrier

\appendix
\section{Computing a Maximum Critical Independent Set}
\label{section:matchings}

\begin{algorithm2e}[!tb]
\DontPrintSemicolon
\SetCommentSty{}
\caption{Larson's algorithm}
\textbf{input} graph $G=(V,E)$\;
\textbf{output} a maximum critical independent set $I$\;
$I \leftarrow \emptyset$ \\
$B(G) \leftarrow$ bi-double of $G$ \\
$V_{\mathrm{remain}} \leftarrow V$ \\
\For{$v$ in $V_{\mathrm{remain}}$}{
 $B_{v}(G) \leftarrow B(G) - N[\{v,v^{\prime}\}]$ \\
 \eIf{$\alpha(B(G)) = \alpha(B_{v}(G)) + 2$}{
    $I \leftarrow I \cup \{v\}$\;
    $V_{\mathrm{remain}} \leftarrow V_{\mathrm{remain}}\setminus N[v]$
 }{
    $V_{\mathrm{remain}} \leftarrow V_{\mathrm{remain}}\setminus \{v\}$
 }
}
\textbf{return} $I$
\label{alg:mcis}
\end{algorithm2e}

We now we describe Larson's algorithm~\cite{larson-2007} for computing a maximum critical independent set, and show that it can be executed in time $O(m^2)$. We assume that the graph is connected, and therefore $m = \Omega(n)$.

Algorithm~\ref{alg:mcis} gives pseudocode for Larson's algorithm. The running time is dominated by the time to compute each of the $O(n)$ independence numbers. We can naively compute each one by computing a bipartite maximum matching in each bi-double graph in $O(n^2m)$ total time, or $O(n^{3/2}m)$ total time with the Hopcroft-Karp algorithm~\cite{hopcroft-karp}.
However, as we now prove, we can reduce this to $O(m^2)$ time, which is more efficient for sparse graphs, where $m = o(n^{3/2})$.
\begin{theorem}
Larson's algorithm can be executed in time $O(m^2)$.
\end{theorem}
\begin{proof}
We begin by computing one maximum matching $M_{B(G)}$ in $B(G)$, and computing the 
independence number $\alpha{(B(G))}$ from $M_{B(G)}$. 
When we compute a new maximum matching $M_{B_v(G)}$ for bi-double graph $B_v({G})$ we use $M_{B(G)}$ as an initial matching. We first create a new matching $M^{\prime}_{B(G)}$ from $M_{B(G)}$, by removing any matched edges
with vertices in $N[\{v,v^{\prime}\}]$, which are removed to create $B_v(G)$. This removes at most $|N(v)| + |N(v^{\prime})|$ edges, as each neighbor can be in at most one matching. Since $M_{B(G)}$ is a maximum matching in $B(G)$, $|M_{B_v(G)}| \leq |M_{B(G)}| \leq |M^{\prime}_{B(G)}| + |N(v)| + |N(v^{\prime})|$, and therefore at most $|N(v)| + |N(v^{\prime})|$ iterations of an augmenting path algorithm are required to compute $M_{B_v(G)}$ from $M^{\prime}_{B(G)}$.

Thus, when evaluating vertex $v$, we can compute the independence number $\alpha(B_v(G))$ in time $O((|N(v)| + |N(v^{\prime})|)(n+m))$. And thus, the running time for the overall algorithm is at most
\begin{align*}
\sum_{v\in V}O((|N(v)| + |N(v^{\prime})|)\cdot(n+m)) = O(m^2).%
\end{align*}\end{proof}
\end{document}

%% file: tables/sanchis.3600.table.tex
\begin{table}[!tb]
\caption{We give the kernel size $k$ and running time $t$ for each reduction technique on synthetically-generated Sanchis data sets. We also list the data used to generate the graphs: the number of vertices $n$, number of edges $m$, and independence number $\alpha(G)$.}
\label{table:kernel-sanchis}
\vspace*{-.5cm}
\scriptsize
\setlength{\tabcolsep}{.7ex}
\begin{center}
\begin{tabular}[tb]{rrr@{\hskip 13pt} rr@{\hskip 13pt} rr@{\hskip 13pt} rr@{\hskip 13pt} rr} 
\toprule
\multicolumn{3}{c}{Graph} & \multicolumn{2}{c}{\critical}& \multicolumn{2}{c}{\maxcritical}& \multicolumn{2}{c}{\advanced}&\multicolumn{2}{c}{\simple} \\
\multicolumn{1}{c}{$n$}  &  \multicolumn{1}{c}{$m$} &  \multicolumn{1}{c}{$\alpha(G)$} & \multicolumn{1}{c}{$k$} & \multicolumn{1}{c}{$t$}    & \multicolumn{1}{c}{$k$} & \multicolumn{1}{c}{$t$} & \multicolumn{1}{c}{$k$} & \multicolumn{1}{c}{$t$} & \multicolumn{1}{c}{$k$} & \multicolumn{1}{c}{$t$}  \\
\cmidrule(r){1-1}\cmidrule(r){2-2}\cmidrule(r){3-3}\cmidrule(r){4-4}\cmidrule(r){5-5}\cmidrule(r){6-6}\cmidrule(r){7-7}\cmidrule(r){8-8}\cmidrule(r){9-9}\cmidrule(r){10-10}\cmidrule{11-11}
\numprint{1000}&\numprint{186723}&\numprint{505}&\textbf{\numprint{0}}&\textbf{\numprint{0.06}}&\textbf{\numprint{0}}&\numprint{0.73}&\textbf{\numprint{0}}&\numprint{0.16}&\numprint{1000}&\numprint{0.00} \\
\numprint{1000}&\numprint{181256}&\numprint{524}&\textbf{\numprint{0}}&\numprint{0.16}&\textbf{\numprint{0}}&\numprint{0.86}&\textbf{\numprint{0}}&\textbf{\numprint{0.15}}&\numprint{1000}&\numprint{0.00} \\
\numprint{2000}&\numprint{711955}&\numprint{1067}&\textbf{\numprint{0}}&\numprint{0.74}&\textbf{\numprint{0}}&\numprint{5.33}&\textbf{\numprint{0}}&\textbf{\numprint{0.37}}&\numprint{2000}&\numprint{0.01} \\
\numprint{2000}&\numprint{686341}&\numprint{1103}&\textbf{\numprint{0}}&\numprint{1.08}&\textbf{\numprint{0}}&\numprint{5.57}&\textbf{\numprint{0}}&\textbf{\numprint{0.34}}&\numprint{2000}&\numprint{0.01} \\
\numprint{3000}&\numprint{536831}&\numprint{1535}&\numprint{2930}&\numprint{0.26}&\textbf{\numprint{0}}&\numprint{1.54}&\textbf{\numprint{0}}&\textbf{\numprint{0.02}}&\textbf{\numprint{0}}&\numprint{0.14} \\
\numprint{3000}&\numprint{513773}&\numprint{1563}&\numprint{2874}&\numprint{0.24}&\textbf{\numprint{0}}&\numprint{1.59}&\textbf{\numprint{0}}&\textbf{\numprint{0.02}}&\textbf{\numprint{0}}&\numprint{0.13} \\
\numprint{4000}&\numprint{929429}&\numprint{2069}&\numprint{3862}&\numprint{0.49}&\textbf{\numprint{0}}&\numprint{2.83}&\textbf{\numprint{0}}&\textbf{\numprint{0.03}}&\textbf{\numprint{0}}&\numprint{0.31} \\
\numprint{4000}&\numprint{805011}&\numprint{2309}&\textbf{\numprint{0}}&\numprint{3.63}&\textbf{\numprint{0}}&\numprint{29.17}&\textbf{\numprint{0}}&\textbf{\numprint{0.47}}&\numprint{4000}&\numprint{0.01} \\
\numprint{5000}&\numprint{1258433}&\numprint{2717}&\numprint{4566}&\numprint{0.70}&\textbf{\numprint{0}}&\numprint{4.85}&\textbf{\numprint{0}}&\textbf{\numprint{0.03}}&\textbf{\numprint{0}}&\numprint{0.50} \\
\numprint{5000}&\numprint{517013}&\numprint{3132}&\textbf{\numprint{0}}&\numprint{3.37}&\textbf{\numprint{0}}&\numprint{25.42}&\textbf{\numprint{0}}&\numprint{0.27}&\textbf{\numprint{0}}&\textbf{\numprint{0.18}} \\
\numprint{6000}&\numprint{1731295}&\numprint{3302}&\numprint{5396}&\numprint{1.02}&\textbf{\numprint{0}}&\numprint{7.06}&\textbf{\numprint{0}}&\textbf{\numprint{0.04}}&\textbf{\numprint{0}}&\numprint{0.78} \\
\numprint{6000}&\numprint{1507280}&\numprint{3412}&\numprint{5176}&\numprint{0.94}&\textbf{\numprint{0}}&\numprint{7.18}&\textbf{\numprint{0}}&\textbf{\numprint{0.04}}&\textbf{\numprint{0}}&\numprint{0.65} \\
\numprint{7000}&\numprint{588713}&\numprint{4493}&\numprint{5014}&\numprint{1.06}&\textbf{\numprint{0}}&\numprint{10.91}&\textbf{\numprint{0}}&\textbf{\numprint{0.03}}&\textbf{\numprint{0}}&\numprint{0.23} \\
\numprint{8000}&\numprint{3099179}&\numprint{4394}&\numprint{7212}&\numprint{1.96}&\textbf{\numprint{0}}&\numprint{12.41}&\textbf{\numprint{0}}&\textbf{\numprint{0.06}}&\textbf{\numprint{0}}&\numprint{1.93} \\
\numprint{8000}&\numprint{428619}&\numprint{5249}&\textbf{\numprint{0}}&\numprint{5.58}&\textbf{\numprint{0}}&\numprint{48.21}&\textbf{\numprint{0}}&\textbf{\numprint{0.17}}&\textbf{\numprint{0}}&\numprint{0.29} \\
\numprint{9000}&\numprint{4040615}&\numprint{4927}&\textbf{\numprint{0}}&\numprint{26.23}&\textbf{\numprint{0}}&\numprint{239.21}&\textbf{\numprint{0}}&\textbf{\numprint{0.74}}&\numprint{9000}&\numprint{0.04} \\
\numprint{9000}&\numprint{451349}&\numprint{5899}&\numprint{6202}&\numprint{1.45}&\textbf{\numprint{0}}&\numprint{18.20}&\textbf{\numprint{0}}&\textbf{\numprint{0.02}}&\textbf{\numprint{0}}&\numprint{0.20} \\
\numprint{10000}&\numprint{4794713}&\numprint{5507}&\numprint{8986}&\numprint{3.02}&\textbf{\numprint{0}}&\numprint{19.86}&\textbf{\numprint{0}}&\textbf{\numprint{0.07}}&\textbf{\numprint{0}}&\numprint{3.78} \\
\numprint{10000}&\numprint{3775385}&\numprint{5811}&\textbf{\numprint{0}}&\numprint{37.28}&\textbf{\numprint{0}}&\numprint{274.31}&\textbf{\numprint{0}}&\textbf{\numprint{1.27}}&\numprint{9994}&\numprint{0.06} \\
\numprint{11000}&\numprint{6344649}&\numprint{5901}&\numprint{10198}&\numprint{4.35}&\textbf{\numprint{0}}&\numprint{23.30}&\textbf{\numprint{0}}&\textbf{\numprint{0.09}}&\textbf{\numprint{0}}&\numprint{5.62} \\
\numprint{11000}&\numprint{2479688}&\numprint{6862}&\textbf{\numprint{0}}&\numprint{33.91}&\textbf{\numprint{0}}&\numprint{223.67}&\textbf{\numprint{0}}&\numprint{1.44}&\textbf{\numprint{0}}&\textbf{\numprint{1.10}} \\
\numprint{12000}&\numprint{5378750}&\numprint{6973}&\numprint{8552}&\numprint{16.46}&\textbf{\numprint{0}}&\numprint{78.69}&\textbf{\numprint{0}}&\textbf{\numprint{0.07}}&\textbf{\numprint{0}}&\numprint{4.55} \\
\numprint{12000}&\numprint{4827152}&\numprint{7098}&\numprint{194}&\numprint{63.80}&\textbf{\numprint{0}}&\numprint{368.03}&\textbf{\numprint{0}}&\textbf{\numprint{0.07}}&\textbf{\numprint{0}}&\numprint{5.35} \\
\numprint{13000}&\numprint{5638263}&\numprint{7698}&\textbf{\numprint{0}}&\numprint{75.59}&\textbf{\numprint{0}}&\numprint{510.73}&\textbf{\numprint{0}}&\numprint{6.56}&\textbf{\numprint{0}}&\textbf{\numprint{4.50}} \\
\numprint{13000}&\numprint{1319528}&\numprint{8474}&\textbf{\numprint{0}}&\numprint{24.54}&\textbf{\numprint{0}}&\numprint{185.99}&\textbf{\numprint{0}}&\textbf{\numprint{0.40}}&\textbf{\numprint{0}}&\numprint{1.15} \\
\numprint{14000}&\numprint{10723774}&\numprint{7417}&\numprint{4}&\numprint{78.11}&\textbf{\numprint{0}}&\numprint{819.25}&\textbf{\numprint{0}}&\textbf{\numprint{25.24}}&\numprint{13880}&\numprint{0.67} \\
\numprint{14000}&\numprint{3250904}&\numprint{8844}&\numprint{10312}&\numprint{4.53}&\textbf{\numprint{0}}&\numprint{47.93}&\textbf{\numprint{0}}&\textbf{\numprint{0.05}}&\textbf{\numprint{0}}&\numprint{2.45} \\
\numprint{15000}&\numprint{6799463}&\numprint{8993}&\numprint{12014}&\numprint{5.82}&\textbf{\numprint{0}}&\numprint{50.81}&\textbf{\numprint{0}}&\textbf{\numprint{0.10}}&\textbf{\numprint{0}}&\numprint{6.58} \\
\numprint{15000}&\numprint{4207335}&\numprint{9413}&\textbf{\numprint{0}}&\numprint{80.07}&\textbf{\numprint{0}}&\numprint{526.90}&\textbf{\numprint{0}}&\textbf{\numprint{2.98}}&\textbf{\numprint{0}}&\numprint{3.37} \\
\numprint{16000}&\numprint{4807361}&\numprint{10042}&\textbf{\numprint{0}}&\numprint{96.05}&\textbf{\numprint{0}}&\numprint{627.97}&\textbf{\numprint{0}}&\numprint{3.76}&\textbf{\numprint{0}}&\textbf{\numprint{3.44}} \\
\numprint{16000}&\numprint{14309249}&\numprint{8401}&\numprint{570}&\numprint{101.08}&\textbf{\numprint{0}}&\numprint{1108.03}&\textbf{\numprint{0}}&\numprint{34.14}&\textbf{\numprint{0}}&\textbf{\numprint{29.35}} \\
\numprint{17000}&\numprint{803659}&\numprint{11239}&\numprint{11522}&\numprint{5.17}&\textbf{\numprint{0}}&\numprint{64.71}&\textbf{\numprint{0}}&\textbf{\numprint{0.03}}&\textbf{\numprint{0}}&\numprint{0.63} \\
\numprint{17000}&\numprint{10662300}&\numprint{9898}&\numprint{14202}&\numprint{7.66}&\textbf{\numprint{0}}&\numprint{60.82}&\textbf{\numprint{0}}&\textbf{\numprint{0.14}}&\textbf{\numprint{0}}&\numprint{12.88} \\
\numprint{18000}&\numprint{5064751}&\numprint{11412}&\numprint{256}&\numprint{124.12}&\textbf{\numprint{0}}&\numprint{683.34}&\textbf{\numprint{0}}&\textbf{\numprint{0.08}}&\textbf{\numprint{0}}&\numprint{6.40} \\
\numprint{18000}&\numprint{1970506}&\numprint{11782}&\numprint{32}&\numprint{53.56}&\textbf{\numprint{0}}&\numprint{372.49}&\textbf{\numprint{0}}&\textbf{\numprint{0.05}}&\textbf{\numprint{0}}&\numprint{1.21} \\
\bottomrule
\end{tabular}
\vspace*{-.5cm}
\end{center}
\end{table}

%% file: tables/erdos-biogrid.3600.table.tex
\begin{table}[!tbh]
\caption{We give the kernel size $k$ and running time $t$ for each reduction technique on Erd\H{o}s and BioGRID graphs. We further give the number of vertices $n$ and edges $m$ for each graph.} 
\label{table:kernel-erdos-biogrid}
\vspace*{-.5cm}
\scriptsize
\setlength{\tabcolsep}{.7ex}
\begin{center}
\begin{tabular}[tb]{lrr@{\hskip 7pt} rr@{\hskip 7pt} rr@{\hskip 7pt} rr@{\hskip 7pt} rr} 
\toprule
\multicolumn{3}{c}{Graph} & \multicolumn{2}{c}{\critical}& \multicolumn{2}{c}{\maxcritical}  & \multicolumn{2}{c}{\advanced}&\multicolumn{2}{c}{\simple} \\
\multicolumn{1}{l}{Name}  & \multicolumn{1}{c}{$n$}  &  \multicolumn{1}{c}{$m$}      & \multicolumn{1}{c}{$k$} & \multicolumn{1}{c}{$t$}    & \multicolumn{1}{c}{$k$} & \multicolumn{1}{c}{$t$} & \multicolumn{1}{c}{$k$} & \multicolumn{1}{c}{$t$} & \multicolumn{1}{c}{$k$} & \multicolumn{1}{c}{$t$}  \\  
\cmidrule(r){1-1}\cmidrule(r){2-2}\cmidrule(r){3-3}\cmidrule(r){4-4}\cmidrule(r){5-5}\cmidrule(r){6-6}\cmidrule(r){7-7}\cmidrule(r){8-8}\cmidrule(r){9-9}\cmidrule(r){10-10}\cmidrule{11-11} \\[2pt]
\multicolumn{11}{l}{\bf{Erd\H{o}s Graphs}} \\[2pt]
erdos971&\numprint{472}&\numprint{1314}&\numprint{350}&\numprint{0.01}&\numprint{124}&\numprint{0.06}&\textbf{\numprint{0}}&\textbf{\numprint{0.00}}&\textbf{\numprint{0}}&\textbf{\numprint{0.00}} \\
erdos972&\numprint{5488}&\numprint{8972}&\numprint{46}&\numprint{0.20}&\textbf{\numprint{0}}&\numprint{13.83}&\textbf{\numprint{0}}&\numprint{0.01}&\textbf{\numprint{0}}&\textbf{\numprint{0.00}} \\
erdos981&\numprint{485}&\numprint{1381}&\numprint{373}&\numprint{0.01}&\numprint{205}&\numprint{0.08}&\textbf{\numprint{0}}&\textbf{\numprint{0.00}}&\textbf{\numprint{0}}&\textbf{\numprint{0.00}} \\
erdos982&\numprint{5822}&\numprint{9505}&\numprint{44}&\numprint{0.22}&\textbf{\numprint{0}}&\numprint{15.65}&\textbf{\numprint{0}}&\numprint{0.01}&\textbf{\numprint{0}}&\textbf{\numprint{0.00}} \\
erdos991&\numprint{492}&\numprint{1417}&\numprint{398}&\numprint{0.01}&\numprint{218}&\numprint{0.07}&\textbf{\numprint{0}}&\numprint{0.01}&\textbf{\numprint{0}}&\textbf{\numprint{0.00}} \\
erdos992&\numprint{6100}&\numprint{9939}&\numprint{42}&\numprint{0.22}&\textbf{\numprint{0}}&\numprint{17.11}&\textbf{\numprint{0}}&\numprint{0.01}&\textbf{\numprint{0}}&\textbf{\numprint{0.00}} \\[6pt]
\multicolumn{11}{l}{\bf{BioGRID Graphs}} \\[2pt]
Arabidopsis-thaliana&\numprint{7225}&\numprint{17223}&\numprint{1534}&\numprint{1.16}&\numprint{188}&\numprint{19.18}&\textbf{\numprint{0}}&\textbf{\numprint{0.02}}&\numprint{31}&\numprint{0.00} \\
Bos-taurus&\numprint{389}&\numprint{357}&\numprint{109}&\numprint{0.01}&\numprint{3}&\numprint{0.05}&\textbf{\numprint{0}}&\textbf{\numprint{0.00}}&\textbf{\numprint{0}}&\textbf{\numprint{0.00}} \\
Caenorhabditis-elegans&\numprint{3974}&\numprint{7918}&\numprint{758}&\numprint{0.30}&\numprint{18}&\numprint{5.38}&\textbf{\numprint{0}}&\numprint{0.01}&\textbf{\numprint{0}}&\textbf{\numprint{0.00}} \\
Candida-albicans-SC5314&\numprint{379}&\numprint{371}&\numprint{66}&\numprint{0.00}&\numprint{14}&\numprint{0.05}&\textbf{\numprint{0}}&\textbf{\numprint{0.00}}&\textbf{\numprint{0}}&\textbf{\numprint{0.00}} \\
Danio-rerio&\numprint{238}&\numprint{249}&\numprint{71}&\numprint{0.00}&\numprint{11}&\numprint{0.02}&\textbf{\numprint{0}}&\textbf{\numprint{0.00}}&\textbf{\numprint{0}}&\textbf{\numprint{0.00}} \\
Drosophila-melanogaster&\numprint{8229}&\numprint{39086}&\numprint{3479}&\numprint{2.78}&\numprint{973}&\numprint{28.01}&\textbf{\numprint{0}}&\textbf{\numprint{0.02}}&\numprint{30}&\numprint{0.01} \\
Escherichia-coli&\numprint{139}&\numprint{122}&\numprint{14}&\numprint{0.00}&\textbf{\numprint{0}}&\numprint{0.01}&\textbf{\numprint{0}}&\textbf{\numprint{0.00}}&\textbf{\numprint{0}}&\textbf{\numprint{0.00}} \\
Gallus-gallus&\numprint{336}&\numprint{343}&\numprint{81}&\numprint{0.00}&\numprint{3}&\numprint{0.04}&\textbf{\numprint{0}}&\textbf{\numprint{0.00}}&\textbf{\numprint{0}}&\textbf{\numprint{0.00}} \\
Hepatitus-C-Virus&\numprint{113}&\numprint{111}&\numprint{2}&\numprint{0.00}&\textbf{\numprint{0}}&\numprint{0.01}&\textbf{\numprint{0}}&\textbf{\numprint{0.00}}&\textbf{\numprint{0}}&\textbf{\numprint{0.00}} \\
Homo-sapiens&\numprint{19592}&\numprint{169285}&\numprint{5675}&\numprint{18.70}&\numprint{1629}&\numprint{210.89}&\textbf{\numprint{0}}&\textbf{\numprint{0.05}}&\numprint{150}&\numprint{0.03} \\
Human-Herpesvirus-1&\numprint{140}&\numprint{140}&\numprint{12}&\numprint{0.00}&\textbf{\numprint{0}}&\numprint{0.01}&\textbf{\numprint{0}}&\textbf{\numprint{0.00}}&\textbf{\numprint{0}}&\textbf{\numprint{0.00}} \\
Human-Herpesvirus-4&\numprint{219}&\numprint{217}&\numprint{2}&\numprint{0.00}&\textbf{\numprint{0}}&\numprint{0.02}&\textbf{\numprint{0}}&\textbf{\numprint{0.00}}&\textbf{\numprint{0}}&\textbf{\numprint{0.00}} \\
Human-Herpesvirus-8&\numprint{137}&\numprint{138}&\textbf{\numprint{0}}&\textbf{\numprint{0.00}}&\textbf{\numprint{0}}&\numprint{0.01}&\textbf{\numprint{0}}&\textbf{\numprint{0.00}}&\textbf{\numprint{0}}&\textbf{\numprint{0.00}} \\
Human-HIV-1&\numprint{1030}&\numprint{1186}&\textbf{\numprint{0}}&\textbf{\numprint{0.00}}&\textbf{\numprint{0}}&\numprint{0.36}&\textbf{\numprint{0}}&\textbf{\numprint{0.00}}&\textbf{\numprint{0}}&\textbf{\numprint{0.00}} \\
Mus-musculus&\numprint{8567}&\numprint{19265}&\numprint{1377}&\numprint{1.39}&\numprint{51}&\numprint{27.56}&\textbf{\numprint{0}}&\numprint{0.01}&\textbf{\numprint{0}}&\textbf{\numprint{0.00}} \\
Oryctolagus-cuniculus&\numprint{183}&\numprint{168}&\numprint{28}&\numprint{0.00}&\textbf{\numprint{0}}&\numprint{0.02}&\textbf{\numprint{0}}&\textbf{\numprint{0.00}}&\textbf{\numprint{0}}&\textbf{\numprint{0.00}} \\
Plasmodium-falciparum-3D7&\numprint{1224}&\numprint{2443}&\numprint{336}&\numprint{0.04}&\textbf{\numprint{0}}&\numprint{0.40}&\textbf{\numprint{0}}&\textbf{\numprint{0.00}}&\textbf{\numprint{0}}&\textbf{\numprint{0.00}} \\
Rattus-norvegicus&\numprint{3066}&\numprint{4139}&\numprint{533}&\numprint{0.14}&\numprint{15}&\numprint{3.04}&\textbf{\numprint{0}}&\textbf{\numprint{0.01}}&\numprint{6}&\numprint{0.00} \\
Saccharomyces-cerevisiae&\numprint{6660}&\numprint{228752}&\numprint{5732}&\numprint{2.48}&\numprint{5180}&\numprint{212.87}&\textbf{\numprint{4086}}&\textbf{\numprint{0.96}}&\numprint{4575}&\numprint{0.05} \\
Schizosaccharomyces-pombe&\numprint{4143}&\numprint{57049}&\numprint{840}&\numprint{1.41}&\numprint{194}&\numprint{11.72}&\textbf{\numprint{0}}&\textbf{\numprint{0.01}}&\textbf{\numprint{0}}&\numprint{0.02} \\
Xenopus-laevis&\numprint{473}&\numprint{520}&\numprint{160}&\numprint{0.01}&\numprint{16}&\numprint{0.06}&\textbf{\numprint{0}}&\textbf{\numprint{0.00}}&\numprint{7}&\numprint{0.00} \\
\bottomrule
\end{tabular}
\vspace*{-.5cm}
\end{center}
\end{table}

%% file: tables/large.networks.3600.table.tex
\begin{table}[!htb]
\caption{We give the size $k_{\max}$ of largest connected component in the kernel from each reduction technique and the running time $t$ of each algorithm to compute an exact maximum independent set. We further give the number of vertices $n$ and edges $m$ for each graph.}
\label{table:solution-large-networks}
\scriptsize
\setlength{\tabcolsep}{.7ex}
\begin{center}
\begin{tabular}{lrr@{\hskip 13pt} rr@{\hskip 13pt} rr} 
\toprule
\multicolumn{3}{c}{Graph}  & \multicolumn{2}{c}{\ai}  &\multicolumn{2}{c}{\simple+\mcs}\\
\multicolumn{1}{l}{Name}  & \multicolumn{1}{c}{$n$}  &  \multicolumn{1}{c}{$m$}      & \multicolumn{1}{c}{$k_{\max}$} & \multicolumn{1}{c}{$t$} & \multicolumn{1}{c}{$k_{\max}$} & \multicolumn{1}{c}{$t$}  \\
\cmidrule(r){1-1}\cmidrule(r){2-2}\cmidrule(r){3-3}\cmidrule(r){4-4}\cmidrule(r){5-5}\cmidrule(r){6-6}\cmidrule(r){7-7} \\[2pt]
\multicolumn{7}{l}{\bf{LAW Graphs}} \\[2pt]
cnr-2000&\numprint{325557}&\numprint{2738969}&\textbf{\numprint{2404}}&-&\numprint{17626}&- \\
dblp-2010&\numprint{326186}&\numprint{807700}&\textbf{\numprint{0}}&\numprint{0.38}&\textbf{\numprint{0}}&\textbf{\numprint{0.14}} \\
dblp-2011&\numprint{986324}&\numprint{3353618}&\textbf{\numprint{0}}&\numprint{1.03}&\numprint{6}&\textbf{\numprint{0.62}} \\
eu-2005&\numprint{862664}&\numprint{16138468}&\textbf{\numprint{51864}}&-&\numprint{313797}&- \\
hollywood-2009&\numprint{1139905}&\numprint{56375711}&\textbf{\numprint{0}}&\numprint{22.01}&\numprint{9}&\textbf{\numprint{21.38}} \\
hollywood-2011&\numprint{2180759}&\numprint{114492816}&\textbf{\numprint{0}}&\numprint{47.50}&\numprint{17}&\textbf{\numprint{44.66}} \\
in-2004&\numprint{1382908}&\numprint{13591473}&\textbf{\numprint{281}}&\textbf{\numprint{4.39}}&\numprint{11615}&- \\
indochina-2004&\numprint{7414866}&\numprint{150984819}&\textbf{\numprint{8246}}&-&\numprint{509355}&- \\
uk-2002&\numprint{18520486}&\numprint{261787258}&\textbf{\numprint{9408}}&-&\numprint{2043389}&- \\[8pt]
\multicolumn{7}{l}{\bf{SNAP Graphs}} \\[2pt]
as-Skitter&\numprint{1696415}&\numprint{11095298}&\textbf{\numprint{597}}&\textbf{\numprint{2111.02}}&\numprint{21174}&- \\
ca-AstroPh&\numprint{18772}&\numprint{198050}&\textbf{\numprint{0}}&\numprint{0.07}&\textbf{\numprint{0}}&\textbf{\numprint{0.02}} \\
ca-CondMat&\numprint{23133}&\numprint{93439}&\textbf{\numprint{0}}&\numprint{0.04}&\textbf{\numprint{0}}&\textbf{\numprint{0.01}} \\
ca-GrQc&\numprint{5242}&\numprint{14484}&\textbf{\numprint{0}}&\numprint{0.04}&\textbf{\numprint{0}}&\textbf{\numprint{0.00}} \\
ca-HepPh&\numprint{12008}&\numprint{118489}&\textbf{\numprint{0}}&\numprint{0.05}&\numprint{7}&\textbf{\numprint{0.01}} \\
ca-HepTh&\numprint{9877}&\numprint{25973}&\textbf{\numprint{0}}&\numprint{0.05}&\textbf{\numprint{0}}&\textbf{\numprint{0.01}} \\
email-Enron&\numprint{36692}&\numprint{183831}&\textbf{\numprint{0}}&\numprint{0.11}&\numprint{6}&\textbf{\numprint{0.03}} \\
email-EuAll&\numprint{265214}&\numprint{364481}&\textbf{\numprint{0}}&\textbf{\numprint{0.09}}&\textbf{\numprint{0}}&\textbf{\numprint{0.09}} \\
p2p-Gnutella04&\numprint{10876}&\numprint{39994}&\textbf{\numprint{0}}&\textbf{\numprint{0.01}}&\numprint{7}&\textbf{\numprint{0.01}} \\
p2p-Gnutella05&\numprint{8846}&\numprint{31839}&\textbf{\numprint{0}}&\textbf{\numprint{0.01}}&\textbf{\numprint{0}}&\textbf{\numprint{0.01}} \\
p2p-Gnutella06&\numprint{8717}&\numprint{31525}&\textbf{\numprint{0}}&\textbf{\numprint{0.01}}&\textbf{\numprint{0}}&\textbf{\numprint{0.01}} \\
p2p-Gnutella08&\numprint{6301}&\numprint{20777}&\textbf{\numprint{0}}&\textbf{\numprint{0.01}}&\textbf{\numprint{0}}&\textbf{\numprint{0.01}} \\
p2p-Gnutella09&\numprint{8114}&\numprint{26013}&\textbf{\numprint{0}}&\numprint{0.02}&\textbf{\numprint{0}}&\textbf{\numprint{0.01}} \\
p2p-Gnutella24&\numprint{26518}&\numprint{65369}&\textbf{\numprint{0}}&\textbf{\numprint{0.02}}&\textbf{\numprint{0}}&\textbf{\numprint{0.02}} \\
p2p-Gnutella25&\numprint{22687}&\numprint{54705}&\textbf{\numprint{0}}&\numprint{0.02}&\textbf{\numprint{0}}&\textbf{\numprint{0.01}} \\
p2p-Gnutella30&\numprint{36682}&\numprint{88328}&\textbf{\numprint{0}}&\textbf{\numprint{0.02}}&\textbf{\numprint{0}}&\textbf{\numprint{0.02}} \\
p2p-Gnutella31&\numprint{62586}&\numprint{147892}&\textbf{\numprint{0}}&\numprint{0.05}&\textbf{\numprint{0}}&\textbf{\numprint{0.03}} \\
roadNet-CA&\numprint{1965206}&\numprint{2766607}&\textbf{\numprint{10807}}&-&\numprint{89667}&- \\
roadNet-PA&\numprint{1088092}&\numprint{1541898}&\textbf{\numprint{5834}}&-&\numprint{35780}&- \\
roadNet-TX&\numprint{1379917}&\numprint{1921660}&\textbf{\numprint{4102}}&-&\numprint{49143}&- \\
soc-Epinions1&\numprint{75879}&\numprint{405740}&\textbf{\numprint{0}}&\numprint{0.07}&\numprint{7}&\textbf{\numprint{0.06}} \\
soc-LiveJournal1&\numprint{4847571}&\numprint{42851237}&\textbf{\numprint{295}}&\textbf{\numprint{8.09}}&\numprint{28037}&- \\
soc-pokec&\numprint{1632803}&\numprint{22301964}&\textbf{\numprint{651503}}&-&\numprint{748755}&- \\
soc-Slashdot0811&\numprint{77360}&\numprint{469180}&\textbf{\numprint{0}}&\textbf{\numprint{0.07}}&\numprint{8}&\numprint{0.13} \\
soc-Slashdot0902&\numprint{82168}&\numprint{504230}&\textbf{\numprint{0}}&\textbf{\numprint{0.11}}&\numprint{15}&\numprint{0.15} \\
web-BerkStan&\numprint{685230}&\numprint{6649470}&\textbf{\numprint{1478}}&\textbf{\numprint{143.43}}&\numprint{62741}&- \\
web-Google&\numprint{875713}&\numprint{4322051}&\textbf{\numprint{70}}&\textbf{\numprint{1.23}}&\numprint{770}&\numprint{1.51} \\
web-NotreDame&\numprint{325729}&\numprint{1090108}&\textbf{\numprint{3548}}&\textbf{\numprint{12.27}}&\numprint{3578}&- \\
web-Stanford&\numprint{281903}&\numprint{1992636}&\textbf{\numprint{2619}}&-&\numprint{10715}&- \\
wiki-Talk&\numprint{2394385}&\numprint{4659565}&\textbf{\numprint{0}}&\textbf{\numprint{0.44}}&\textbf{\numprint{0}}&\numprint{2.32} \\
wiki-Vote&\numprint{7115}&\numprint{100762}&\textbf{\numprint{0}}&\textbf{\numprint{0.01}}&\textbf{\numprint{0}}&\numprint{0.02} \\[8pt]
\multicolumn{7}{l}{\bf{KONECT Graphs}} \\[2pt]
flickr-growth&\numprint{2302925}&\numprint{22838276}&\textbf{\numprint{9}}&\textbf{\numprint{1.60}}&\numprint{139}&\numprint{31.59} \\
flickr-links&\numprint{1715255}&\numprint{15555041}&\textbf{\numprint{9}}&\textbf{\numprint{1.10}}&\numprint{68}&\numprint{17.04} \\
libimseti&\numprint{220970}&\numprint{17233144}&\textbf{\numprint{49399}}&\textbf{\numprint{1371.18}}&\numprint{141008}&- \\
orkut-links&\numprint{3072441}&\numprint{117185083}&\textbf{\numprint{2545612}}&-&\numprint{2701058}&- \\
petster-carnivore&\numprint{623766}&\numprint{15695166}&\textbf{\numprint{0}}&\textbf{\numprint{2.50}}&\numprint{117}&\numprint{4.77} \\
petster-cat&\numprint{149700}&\numprint{5448197}&\textbf{\numprint{66}}&\textbf{\numprint{2.83}}&\numprint{68152}&- \\
petster-dog&\numprint{426820}&\numprint{8543549}&\textbf{\numprint{231}}&\textbf{\numprint{4.59}}&\numprint{139270}&- \\
youtube-links&\numprint{1138499}&\numprint{2990443}&\textbf{\numprint{0}}&\textbf{\numprint{0.43}}&\numprint{19}&\numprint{3.12} \\
youtube-u-growth&\numprint{3223643}&\numprint{9376594}&\textbf{\numprint{0}}&\textbf{\numprint{1.51}}&\numprint{33}&\numprint{17.59} \\
baidu-internallink&\numprint{2141300}&\numprint{17014946}&\textbf{\numprint{10}}&\textbf{\numprint{1.02}}&\numprint{71}&\numprint{36.99} \\
baidu-relatedpages&\numprint{415641}&\numprint{2374044}&\textbf{\numprint{492}}&\textbf{\numprint{1.59}}&\numprint{11458}&- \\
hudong-internallink&\numprint{1984484}&\numprint{14428382}&\textbf{\numprint{79}}&\textbf{\numprint{1.89}}&\numprint{1546}&\numprint{45.92} \\
\hline
\end{tabular}
\end{center}
\end{table}